\def\counterwithin#1#2{} 
\documentclass[a4paper,USenglish
,final
]{eurocg25}

\usepackage[disable] 
{todonotes}

\newif\iflong
\longtrue 

\newif\ifverylong
 \verylongtrue 

\newif\ifarXiv
\arXivtrue 

\ifarXiv
\verylongtrue
\makeatletter
\def\copyrightline{%
  \scriptsize
  \vtop{\hsize\textwidth
    \nobreakspace\\
    \ifx\@EventLongTitle\@empty\else
    \@EventLongTitle.\\\fi
 This is the full version of a paper that will be 
 presented 
 at
    \@EventShortTitle.  }}

\let\@oddfoot\@empty 
\makeatother
\fi

\usepackage{graphicx} 
\usepackage{amsmath,paralist}
\usepackage{hyperref}
\usepackage{xcolor}

\definecolor{codegreen}{rgb}{0,0.6,0}
\definecolor{codegray}{rgb}{0.5,0.5,0.5}
\definecolor{codepurple}{rgb}{0.58,0,0.82}
\definecolor{codecyan}{rgb}{0.,0.0,0.5}
\definecolor{backcolour}{rgb}{0.95,0.95,0.92}
\usepackage{listings}
\addtocounter{lstnumber}{-1}

\newcommand{\defn}[1]{\textit{\textbf{\boldmath #1}}}

\let\paragraph=\subparagraph 

\usepackage{aliascnt}
\makeatletter
\let\corollary\@undefined
\let\endcorollary\@undefined
\let\lemma\@undefined
\let\endlemma\@undefined
\let\obs\@undefined
\let\endobs\@undefined
\makeatother
\theoremstyle{plain}
\newaliascnt{lemma}{theorem}
\newtheorem{lemma}[lemma]{Lemma}
\aliascntresetthe{lemma}

\newaliascnt{corollary}{theorem}
\newtheorem{corollary}[corollary]{Corollary}
\aliascntresetthe{corollary}

\newaliascnt{obs}{theorem}
\newtheorem{obs}[obs]{Observation}
\aliascntresetthe{obs}

\title{Minimum spanning blob-trees} 

\author[1]{Katharina Klost}
\author[2]{Marc van Kreveld}
\author[3]{Daniel Perz}
\author[1]{Günter Rote}
\author[4]{Josef Tkadlec}
\affil[1]{Freie Universität Berlin, Institut für Informatik\\
  \texttt{\{kathklost,rote\}@inf.fu-berlin.de}}
\affil[2]{Department of Information and Computing
        Sciences, Utrecht University\\
  \texttt{m.j.vankreveld@uu.nl}}
\affil[3]{University of Perugia\\
  \texttt{daniel.perz@unipg.it}}
\affil[4]{Charles University\\
  \texttt{josef.tkadlec@iuuk.mff.cuni.cz}}
\authorrunning{K. Klost,  M. van Kreveld, D. Perz, G. Rote, and J. Tkadlec}
\ArticleNo{25}

\begin{document}
{
}

\maketitle

\begin{abstract}
  We investigate blob-trees,
  a new way of connecting a set of points, by a mixture
  of enclosing them by cycles (as in the convex hull) and
  connecting them by edges (as in a spanning tree).
  We show that a minimum-cost blob-tree
  for $n$ points can be
  computed in $O(n^3)$ time.
\end{abstract}


\section{Introduction}

Any introductory course on computational geometry will treat convex hulls and minimum spanning trees (MSTs) for a set $P$ of $n$ points in the plane. These are the least-cost structures that enclose or connect the points.
In this paper we investigate a new structure, the \emph{blob-tree}, that combines the ideas of enclosing and connecting. 
A \emph{blob} is a simple polygon, 
and all points of $P$ enclosed in a blob are considered to be connected to each other.
A \emph{tree-edge} is a segment between two points of $P$, and it connects its
two endpoints.
A blob-tree is a collection of blobs and tree-edges that collectively connects all of $P$,
see \autoref{fig:example-solution}.
\iflong
The convex hull and the MST of $P$ are special cases of a blob-tree.
\fi

We are looking for the blob-tree that minimizes the total length of all edges drawn: the perimeter of the blobs plus the tree-edges.
\iflong As we show in 
\else
By
\fi
Lemma~\ref{lem:convex} below, 
\iflong the \fi blobs are convex and disjoint.
\iflong
Morever, when contracting the blobs to vertices, the tree-edges form a
tree, so the name ``blob-tree'' is warranted. 
\else
\looseness-1
\fi
\todo[inline]{Pepa: consider commenting more generally on Minimum Spanning substructure in a hypergraph with subadditive weights on the hyperedges? In our case, the cost of a $k$-tuple with $k\ge 3$ is the length of the convex hull and the cost of a $2$-tuple is the edge length. There is work on minimum spanning substructures on hypergraphs that we could cite, e.g. (somewhat random): ``Approximating Minimum Spanning Sets in Hypergraphs and Polymatroids'' or ``Minimum Tree Supports for Hypergraphs and
Low-Concurrency Euler Diagrams''.}

\begin{figure}
    \centering
    \hskip-1,1cm
    \includegraphics{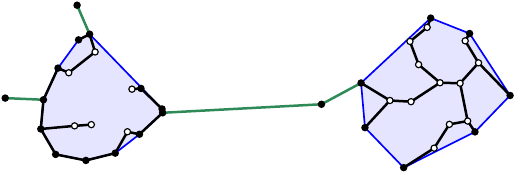}
    \qquad
    \includegraphics{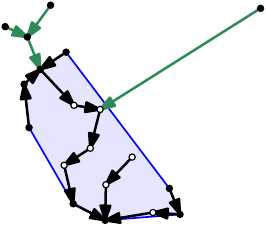}
    \caption{The optimal blob-trees on two point sets (computed by our implementation \cite{repo}). Blob edges are blue, tree-edges are green.
\iflong    
    The cost of the solution is the total length of blue and green. 
    \fi
    MST edges are black (and green, if they are part of the blob-tree).}
    \label{fig:example-solution}
\end{figure}

While blob-trees appear interesting in their own right,  similar
structures have been used in information visualization, in particular
in KelpFusion~\cite{mrsad-2013kelpfusion}. This visualization style is
based on \emph{shortest-path graphs}~\cite{de2011delineating}, a
one-parameter family of connected structures whose extremes are the
convex hull and the MST, 
%
without regard to a specific optimization criterion
\ifverylong
More precisely,
the shortest path graph for a parameter $t\in [1,\infty)$ uses as the
distance between two points $p$ and $q$ the Euclidean distance to the
power~$t$. The shortest path graph is the union of all shortest paths
between pairs of points; for $t=1$ this is the complete graph and for
$t\rightarrow \infty$ this is the Euclidean MST. For KelpFusion, this
graph is then converted into a connected set of edges (not
necessarily convex blobs) by filling in faces.
\fi
Other related work considers placing shortest fences for subsets of points~\cite{aglr-gmsfs-20,arkin1993geometric,bcflr-fscsf-25}.

As seen in Figures~\ref{fig:example-solution} and~\ref{fig:versions}, an optimal solution might involve a tree-edge that crosses a boundary of a blob. 
If we disallow 
such crossings (version (ii)), the optimal solution could get longer and might include non-convex blobs.
If we further required the blobs to be convex (version (iii)), the
optimal solution might become even longer, see
Figure~\ref{fig:versions}.
\ifverylong
A~different version ignores the parts of tree-edges that lie inside a
blob when computing the length of a solution.
This is less natural. The solution might choose edges that run deep
into the blob, just to get a shorter distance to the blob boundary.
\fi


\begin{figure}
    \centering
    \includegraphics{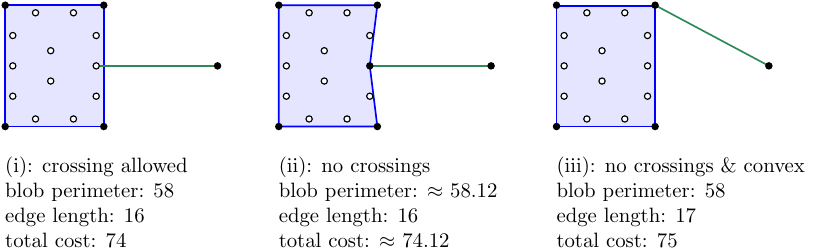}
    \caption{Three variants of blob-trees for the same point set, and their costs.}
    \label{fig:versions}
\end{figure}

\iflong In this paper, w\else W\fi e concentrate on the original problem.
It turns out that the optimal solution can be obtained from the MST by replacing several of its subtrees by  blobs
(\autoref{lem:struct-MST});
the other two variants (ii) and (iii) do not have this property.
As our main result we show that the optimal blob-tree can be computed in $O(n^3)$ time, and we implement the algorithm in \textsc{Python}~\cite{repo}.
\iflong Efficient c\else C\fi omputation of optimal blob-trees in variants (ii) or (iii)
\iflong is still 
\else remains \looseness-1
\fi
open.



Our algorithm is based on dynamic programming and builds upon ideas from \cite{eorw-fmakg-DCG92
}. That paper shows how to compute a single optimal convex polygon (or blob) in a point set that has certain properties, using dynamic programming. For example, it can compute the smallest area or perimeter convex polygon with $k$ points on the boundary, or the smallest area convex polygon containing at least $k$ points inside or on the boundary in $O(kn^3)$ time.
We show how to deal with multiple blobs and with the minimum spanning connecting structure to compute an overall minimum solution.
\iflong

\fi
All proofs that are not
\ifarXiv
 \iflong given \fi
 in the main part are in
the appendix.
\else given here are in the full version \cite{fullversion}.
\fi 
\iflong\else\looseness-1\fi


\section{Preliminaries}
We assume that the \iflong point \fi
set $P$ is in general position.
In particular,
\iflong we assume that \fi the pairwise distances between points are unique and no three points are on a line.
This implies that the MST of the point set is 
unique. We also assume that no two points have the same $x$- or $y$-coordinate.


\begin{lemma}
\label{lem:convex}
    In any optimal solution: all blobs are convex polygons with vertices at $P$;
    any two blobs are disjoint; 
    and when contracting the blobs to vertices, the tree-edges form a tree.
\end{lemma}
\begin{proof}
    The claims follows from the triangle inequality.
    First, we can replace any blob that does not have the claimed property 
    with the convex hull
    of the points of $P$ contained in it, 
    creating a blob-tree with a strictly smaller total cost.
%
    Similarly, two intersecting 
  blobs $B_1$ and $B_2$ can be replaced by the convex hull of the union $B_1\cup B_2$.
    
  Contraction of blobs creates an abstract graph $G$ in which every
  blob and every singleton point that is not in a blob is represented
  by a vertex.
  For every tree-edge, $G$ contains an edge between the corresponding
  vertices.  Since the blob-tree is connected and spans all points,
  $G$ is connected. If $G$ contained a cycle, then any tree-edge
  corresponding to an edge from such a cycle could be removed from the
  blob-tree. So the graph $G$ is a tree.
%
\end{proof}

Throughout the paper, $T$ will denote the minimum spanning tree (MST) of~$P$.
We select the lowest point $r$
as the root and direct all edges of $T$ towards
~$r$. For a node $u$, we denote its parent in $T$ by $u'$ and the directed edge from \(u\) to its parent by $uu'$.

\section{Structural insights}
The following key lemma with its immediate corollary shows that the optimal solution can be obtained from the MST by replacing some of its subtrees by the respective blobs.

\begin{lemma}\label{lem:struct-MST}
Let $T$ be the MST. Then, in any optimal blob-tree $S$,
every blob is a convex hull of some subtree of $T$ and
every tree-edge is an edge of $T$. 
\end{lemma}


\iflong
Lemmas~\ref{lem:convex} and~\ref{lem:struct-MST} immediately yield the following corollary.
\fi
\begin{corollary}
\label{coro:connected-inside-blob}
    Let $B$ be a blob in any optimal blob-tree $S$. Let $T$ be the MST. Then:
\begin{enumerate}[\rm(a)]
\item
A path in $T$ that leaves $B$ never comes back to $B$.
    \item 
The subgraph of $T$ induced by the points in $B$ is connected.
\item
No MST edge with both endpoints outside $B$ intersects $B$.
\qed
\end{enumerate}
\end{corollary}




\begin{figure}
\centering
\includegraphics{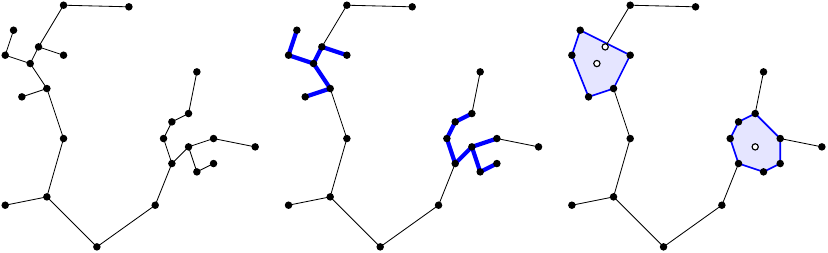}
\caption{\iflong A minimum spanning tree,
\else An MST, \fi a potential subset of
non-tree edges (fat blue),
\iftrue and \fi the resulting blob-tree.}
\label{fig:blob-intuition}
\end{figure}

\ifverylong 
\autoref{lem:struct-MST} and \autoref{coro:connected-inside-blob} are the foundation
of 
 a straightforward exponential-time 
algorithm for computing the optimum blob-tree, see \autoref{fig:blob-intuition}.
After computing the MST $T$, we \emph{guess} which of the $n-1$ MST edges are tree-edges (the black edges in the middle of \autoref{fig:blob-intuition}) and which edges are inside blobs (the blue edges). 
Removing the black edges cuts the tree into blue components.
We take the convex hull of each blue component as a blob and add the black edges,
as shown in the right of \autoref{fig:blob-intuition}.
In total, there are $2^{n-1}$ candidate solutions, and we take the one that minimizes the total cost. 
\fi 

The structure of optimal blob-trees can be used to develop a polynomial-time algorithm by dynamic programming. We extend the structure in two simple ways.

\begin{figure}
\centering
\includegraphics{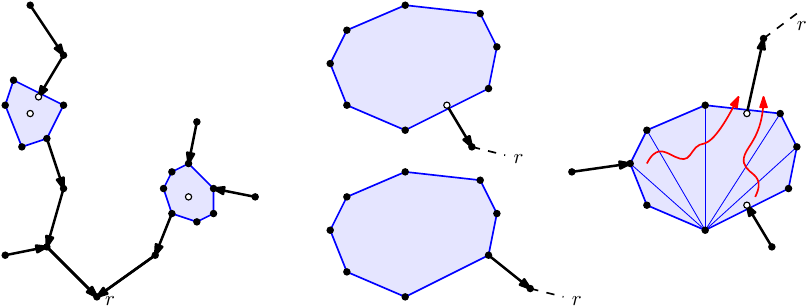}
\caption{A rooted, directed blob tree, two types of exits, and a bottom-vertex triangulation of a blob with a dynamic programming order through it towards the root.}
\label{fig:blob-intuition2}
\end{figure}

Firstly, we use the direction of all edges of $T$ towards the root $r$. If we were to contract each blob to a single point, then we would have a normal rooted tree. Consequently, every blob that does not contain the root has one unique exit edge and every blob has zero or more entry edges. The exit from a blob towards the root can be via an edge from one of its hull vertices or via an edge from an interior blob point that crosses a hull edge. Similarly, blob entrances come in two types, see \autoref{fig:blob-intuition2}.

Secondly, we consider a bottom-vertex triangulation of each blob,
where
the lowest point is connected to all other vertices on the boundary of the blob. We call the diagonals of the triangulation \emph{chords} and the exterior edges \emph{walls}.
The chords and triangles serve to have simple (constant size) elements to be used to define substructures in the dynamic program. Optimal solutions to smaller subproblems are ``transported'' through a blob from any entrance to the unique exit, using the chords in between, to build solutions to larger subproblems.
\iflong\else\looseness-1\fi

There will be three types of subproblems: (i) edge subproblems, for
each MST edge that is used as a blob-tree edge,
(ii) chord subproblems, and (iii) wall subproblems. There are $n-1$ of the first type 
and $O(n^2)$ of the latter two types.
Walls and chords do not have an implied orientation (unlike the edges
of~$T$). Each segment must be considered as a potential wall or chord
separately in both orientations. 

For uniformity, we draw chords from the
bottom vertex of a blob also to its neighbors on the blob boundary,
cutting off two degenerate \emph{digons} from the blob,
see \autoref{fig:parition}.
\todo{Marc: This is not really intuition anymore.}
Thus,
a 
 blob with $k$ boundary points has $k-1$ chords and $k$ walls,
 and it is cut into 
 $k-2$ triangles, each bounded by two chords and one wall, and two digons, bounded by 
 one chord and one wall.
 \iflong\else\looseness-1\fi

Every chord has a \emph{forward side}, towards the tree root,
and a \emph{backward side}; both are defined later in detail.
\todo{G. Intuition again}
Our algorithm 
proceeds from
the leaves towards the root, as indicated by the red arrows in
\autoref{fig:blob-intuition2}.
Thus, the subproblem associated to a chord consists of all points that are on the backward side,
including the points that are in other blobs hanging off
the current blob in the blob tree. It turns out that
the points that are involved in this subproblem can be
uniquely identified with the help of the MST, without
knowing the optimum blob tree.


\begin{figure}
    \centering
    \includegraphics{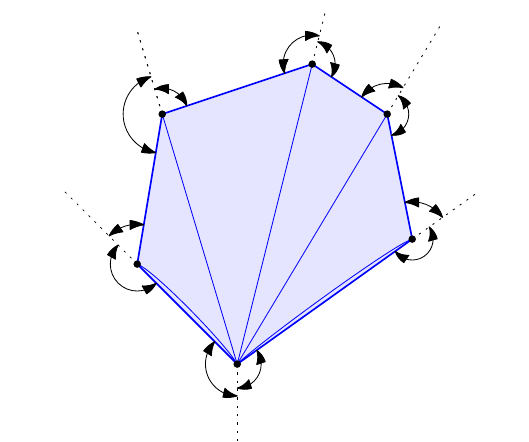}
    \caption{A blob is divided into triangles and digons.}
    \label{fig:parition}
\end{figure}

\begin{figure}
    \centering

\ifarXiv    
\tabskip = 0pt plus 0.4 fil
\halign to \textwidth{\hfil#\hfil\tabskip = 0pt plus 1 fil
&\hfil#\hfil
&\hfil#\hfil\tabskip = 0pt plus 0.4 fil
\cr
    \includegraphics[page=1]{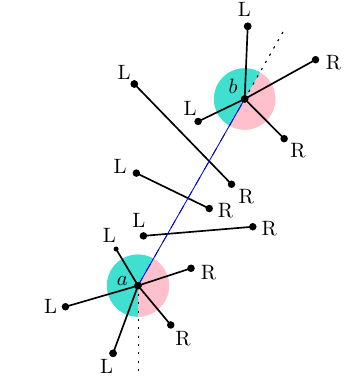}%
&%
    {\hskip-8mm
    \includegraphics[page=2]{classify_chord_triangles.pdf}%
    \hskip-8mm}%
\cr
(a)&(b)\cr
}
\else
\includegraphics[page=1]{classify_chord_triangles.pdf}%
\fi
    \caption{\ifarXiv(a) \fi Labeling endpoints of MST-edges $X^+_{ab}$
    crossing $ab$ or incident to $a$ or $b$ as left (L) and right (R). 
    The neighbors of 
    $a$ and $b$ are classified as right (R)
    with respect to $ab$ if the edges emanate in the pink region, otherwise they are
    left (L).
    \iflong
    (The shown edges cannot all be MST edges simultaneously.)
    \fi
\ifarXiv    (b)~
    Proof of \autoref{lem:enpoints_r_or_l}.\fi
    }
    \label{fig:left-right}
\end{figure}
\begin{figure}
    \centering
  \includegraphics[page=4]{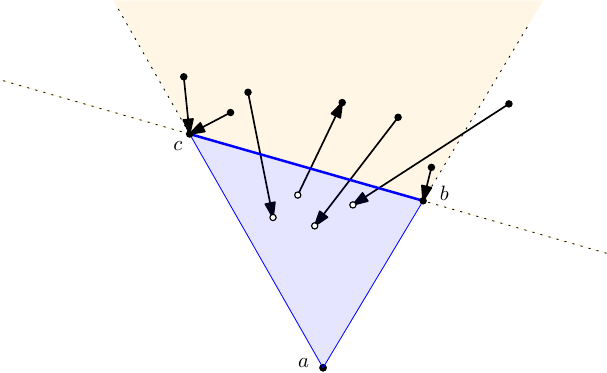}%

    \caption{Six entry edges and one exit edge for the triangle $abc$.
 \iflong   The four edges that cross $bc$ are associated to the wall $bc$. \fi
sd  }
    \label{fig:classify_sides}
\end{figure}
  

Let \(a,b\) be a pair of points. 
Let \(X_{ab}\) be the set of edges of \(T\) that cross the line segment \(ab\), and
let $X^+_{ab}$ denote the edges in  \(X_{ab}\) plus the edges of \(T\) that are incident to \(a\) or~\(b\). 

\subparagraph{Sidedness.} 
Assume that \(a\) is below \(b\) and consider the \emph{boundary curve} that is obtained from the union of the vertical ray below \(a\) and the ray from \(a\) through \(b\), see \autoref{fig:left-right}a.
This curve subdivides the plane into a left and a right side. 
We call a point \(v\) that is a vertex of an edge in \(X_{ab}\) a \emph{left} or \emph{right endpoint}, depending on if it lies
left or right of
the boundary curve. 
The 
endpoints of edges in \(X^+_{ab}\setminus X_{ab}\) that are not \(a\) or \(b\) 
are classified analogous.
\iflong \else\looseness-1\fi


\begin{lemma}\label{lem:enpoints_r_or_l}
    Every edge in \(X_{ab}\) has exactly one left and one right endpoint.
\end{lemma}


\paragraph{Valid chords.}
Removing \(X^+_{ab}\) from \(T\) gives a forest.
We call the segment \(ab\) a \emph{valid chord} if no component in \(T\setminus X^+_{ab}\) contains both a left and a right endpoint.
For a valid chord, the components of \(T\setminus X^+_{ab}\) can be partitioned into \emph{left} and \emph{right components}, depending on the characterization of the vertices from \(X^+_{ab}\) contained in the component.
\iflong
Note that this partitions \(P\setminus \{a,b\}\).
\fi

The side assigned to the root \(r\) is the \emph{forward} side of the chord and the other side the \emph{backward} side. 
(If the lower point $a$ of the chord is $r$, we arbitrarily declare the right side as the forward side. 
This convention simplifies the treatment of the ``root blob''.)
If the forward side is the right side, then \(ab\) is a \defn{right-facing} chord, otherwise it is a \defn{left-facing} chord.
\iflong\else\looseness-1\fi

\todo[inline]{Note that validity and the sidedness of a chord can be determined without reference to a solution. ?]}

\begin{lemma}
\label{lem:valid}
All chords in an optimal solution are valid. 
 \end{lemma}


\subparagraph{Walls.}
Walls are the edges on the boundary of a blob in a counterclockwise traversal.
Let \(b,c\) be a pair of points.
Then \(uu'\in X_{bc}\) is an \defn{entry edge} for \(\overrightarrow{bc}\) if it crosses
\(\overrightarrow{bc}\) from right to left (from outside the blob to inside)
and an \defn{exit edge} otherwise. 

\subparagraph{Triangles and Digons.}
Among the non-root triangles in to which a blob is decomposed from the lowest vertex, there is a unique wall, where the edge towards the root exits. If this wall is part of a triangle, it is an \emph{exit triangle} the other triangles are
LR-triangles (left-to-right), RL-triangles (right-to-left), depending
on the direction in which the root lies,
see \autoref{fig:exit}.
Points adjacent to vertices of the blob are assigned to the triangles or digons according to the extension of the triangles by the rays through \(a\) and the other two vertices.

\begin{figure}
    \centering
    \includegraphics[page=3]{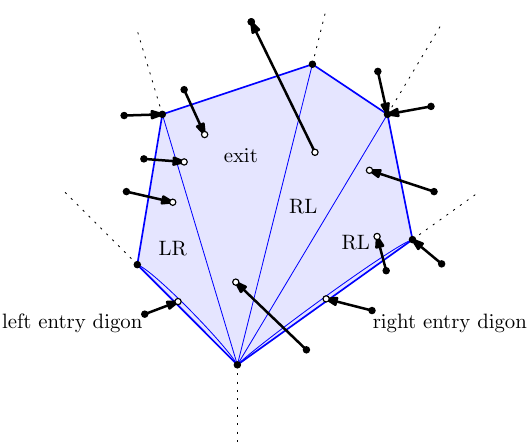}
    \caption{An example with one exit triangle, several LR- and RL-triangles, and two entry digons.}
    \label{fig:exit}
\end{figure}

For every 
triangle $abc$ with lowest point $a$,
we can, by analyzing the tree edges that cross the sides or are incident to the vertices $a,b,v$, classify it as a potential
LR-triangle, RL-triangle or exit triangle. A similar classification is obtained for digons.
The details are given in
\ifarXiv
\autoref{sec:classification}.
\else \cite[Appendix~D]{fullversion}.
\fi 
\todo[inline]{THE FOLLOWING REMARK MIGHT BE NICE IN APPROPRIATE PLACE: The trick of assigning $a$ to the right, if $a$ is the root leads to the fact that there are no root triangles and only root digons are needed. }
\iflong

\fi
The following lemma shows that this concept of valid triangles and digons agrees with the intended meaning:
 
\begin{lemma}
\label{lem:exit}
Let \(B\) be a blob in an optimal solution. Then, 
    \begin{enumerate}[(a)]
        \item There is exactly one exit triangle or digon.
        \todo
        [inline]{or root digon. Should we ARTIFICIALLY declare the root digon to be the exit digon?\\
        K:Currently we are inconsistent with this, see todos in the appendix}
        \item Every other triangle and digon in the decomposition of the blob from the lowest point is a valid LR-triangle, RL-triangle or a valid entry digon.
    \end{enumerate}
\end{lemma}

\section{The algorithm}

The algorithm does not actually check that the blobs that we form are convex.
The only property that we implicitly enforce is that each blob is star-shaped around
the lowest vertex $a$.
\iflong
Thus, the algorithm 
considers also nonconvex blobs for potential solutions.
However, we know from \autoref{lem:convex} that such solutions cannot be optimal.
\else
 We know from \autoref{lem:convex} that nonconvex solutions cannot be optimal.
\fi
As mentioned, we consider three type of subproblems.

\subparagraph{{Edge} problems.}
For each MST edge \(uu'\), we define a problem $\texttt{edge}[u]$, for the
optimum solution in which $uu'$ is used as a tree-edge.
In addition, for the root $r$, we have
$\texttt{edge}[r]$ for the overall problem.
We define \(V_u\) as the set of points in the subtree of  \(T\) rooted at \(u\).
\todo{Maybe $T_u$ instead of $V_u$?}
We define the \emph{size} of the problem as the cardinality of $V_u$.

\subparagraph{Chord problems.}
For each valid chord $ab$, we define a problem
\(\texttt{chord}[a,b]\) for the optimum solution on the backward side of $ab$,
supposing that $ab$ occurs as a chord of a blob in the solution.
We define \(V_{ab}\) as the set of all points that are in a backward component of \(T\setminus X^+_{ab}\).
We define the \emph{size} of the problem \(\texttt{chord}[a,b]\) as the cardinality of $V_{ab}$.
\iflong
For clarity, we will sometimes write  $V_{ab}$
as $\overleftarrow V\!_{\!ab}$ for a right-facing chord and
\todo{Is this really a good notation idea?}
as $\overrightarrow V\!_{\!ab}$ for a left-facing chord.
\fi

\paragraph{Wall problems.}
We denote by \(W_{abc}\) the set of endpoint of entry edges for \(abc\) that lie outside of the blob.
We will frequently need the values
$\sum_{u\in W_{abc}} \texttt{edge}[u]$.
We split this 
into
\begin{displaymath}
\sum_{u\in W_{abc}} \texttt{edge}[u] =
\sum_{u\in W_{bc}} \texttt{edge}[u]+ 
\sum_{u\in W_{abc}, \ uu' \text{ entry edge into \(b\) or \(c\)}}\texttt{edge}[u]
\end{displaymath}
To speed up the evaluation of this term, the first sum is stored as the solution
of an auxiliary subproblem, $\texttt{wall}[bc]$;
the second sum can be
computed in constant time,
\iflong
because
there are only a constant number of entry edges incident to $b$ and $c$, 
\fi 
because the
degree of a Euclidean MST is bounded by~6.
    
%
\iflong
For each pair of points $b,c$,
we denote by \(W_{bc}\) the set of points \(w\) where \(ww' \in X_{bc}\) is an entry edge for \(bc\).
The size of this problem is defined as the sum of the sizes of the problems $\texttt{edge}[u]$. 
If $bc$ is a wall in an optimal solution, the subtrees rooted at the nodes $w\in W_{bc}$ are disjoint.
However, we don't check this condition.
Hence the size can be bigger than $n$.
\fi

\paragraph{Relation between subproblems.}
In
\ifarXiv
\autoref{sec:subproblem_structure_digon},
\else \cite[Appendix~G, Lemmas 9 and~10]{fullversion},
\fi 
we describe the relations between the sets of points that define the
subproblems\ifarXiv
\ (Lemmas~\ref{lem:subproblem_structure_triangle} and~\ref{lem:subproblem_structure_digon})\fi.
These lemmas ensure that the solution
of every subproblem depends only on problems of smaller size.

\paragraph{Preprocessing.}
In a preprocessing phase, we determine the size of each subproblem,
in order to know the order in which we have to solve the subproblems.
In
\ifarXiv
\autoref{sec:preprocessing},
\else \cite[Appendix~H, Lemma~11]{fullversion},
\fi 
we give details about
the extra information that is gathered in the preprocessing step\ifarXiv
\ (\autoref{lem:preprocessing})\fi.



\subparagraph{{Chord} problems.}
For each valid chord $ab$, we determine its sidedness.
If it is right-facing, we consider all possibilities for
an LR-triangle $abc$, as well as the possibility
that $ab$ is a left entering digon.
Left-facing chords are analogous. More details are in
\ifarXiv
\autoref{sec:solution}.
\else \cite[Appendix~I]{fullversion}.
\fi 

\subparagraph{{Edge} problems.}
For each MST edge \(uu'\), we have two possibilities.
If $u$ is not in a blob, all incoming MST edges of $u$ must
be tree-edges, and we can accumulate the values of the corresponding
subproblems.
If $u$ is in a blob, we consider all potential exit triangles $abc$ for which the $uu'$ is the exiting edge crossing the wall $bc$,
as well as the analogous possibility of an exit digon.
Details are given in
\ifarXiv
\autoref{sec:solution}.
\else \cite[Appendix~I]{fullversion}.
\fi

\begin{theorem}
\label{thm:algo}
    The dynamic program \iflong above correctly \fi solves the minimum blob-tree problem in \(O(n^3)\) time.
\end{theorem}
\begin{proof}
For each size, we first solve the \texttt{edge} problems, then the \texttt{wall} problems, and then the \texttt{chord} problems.
\ifarXiv
    Lemmas~\ref{lem:subproblem_structure_triangle} and~\ref{lem:subproblem_structure_digon}
    in
\autoref{sec:subproblem_structure_digon}
\else
Lemmas 9 and~10 in
\cite[Appendix~G]{fullversion}
\fi 
    show that every subproblem
 solution that is needed is already computed,
 and correctness follows from these lemmas.

The straightforward running time analysis is given
in
\ifarXiv
\autoref{sec:runtime}.
\else \cite[Appendix~J]{fullversion}.
\fi
\end{proof}


{\bf Acknowledgments.} We thank
Bettina Speckmann for co-proposing this problem, and Philipp Kindermann for hosting GG Week 2024 in Trier, where this research started. 

\bibliographystyle{plainurl}
\bibliography{blobtrees,enclosing}

\ifarXiv\else
\end{document}
\fi

\appendix
\section{Proof of \autoref{lem:struct-MST}}
\label{sec:proof:struct-MST}
\begin{proof}
First, consider an MST-edge \(e=u'u''
\) that is not a tree-edge or inside a blob in a solution~$S$.
We show how to construct a better solution $\hat S$, and therefore $S$ cannot be optimal.
The removal of \(e\) splits \(T\) into two subtrees \(T'\) and \(T''\)
with $u'\in T'$ and $u''\in T''$.
There has to be a connection between \(u'\) and \(u''\) in~$S$, 
consisting of an alternating sequence of paths and blobs,
see \autoref{fig:struct-MST} for an illustration. 
More precisely,
there are $k\ge1$ paths $u'=v_0\to w_0$, $v_1\to w_1$, \ldots, $v_{k}\to w_k=u''$ consisting
of tree-edges, such that $w_i$ and $v_{i+1}$ are in a common blob.
It can happen that the first or last path is trivial, i.\,e., $u'=w_0$ or $v_k=u''$, but there must be at least one non-trivial path, because otherwise $u'$ and $u''$ are in the same blob,
contrary to our assumption.

At some point in this sequence, there must be a switch from $T'$ to $T''$. If this switch
occurs at an edge $\bar e$ of one of the paths, we can exchange $\bar e$ by $e$ and reduce
the cost.
Suppose that the switch occurs in one of the blobs, say between
 $w_i\in T'$ and $v_{i+1}\in T''$. Then we know that the distance between $w_i$ and $v_{i+1}$
 is larger than the length of $e$.
In this case, we obtain the shorter solution $\hat S$ by forming a single blob
from the convex hull of all paths and blobs on the way between $u'$ and $u''$.
To see that this is shorter than $S$, we first form the outer contour $K$ of all paths and blobs
together with the edge $e$. We have to pay the length of $e$, but we save the parts of the blob perimeter that lie inside the contour.
In particular, consider the last edge
$xw_i$ of the path
$v_i\to w_i$ and the first edge  $v_{i+1}y$ of the path $v_{i+1}\to w_{i+1}$.
Let $\bar x$ and $\bar y$ be the intersections of the edges 
 $xw_i$ and $v_{i+1}y'$ with the boundary of $B$. (It is possible that
 $\bar x=w_i$ or $\bar y=v_{i+1}$.)
When comparing the blob 
with the outer contour $K$, 
we save
at least the parts $w_i\bar x$ and $\bar yw_{i+1}$ of the edges plus
a part of the boundary of $B$ from  $\bar x$ to $\bar y$,
forming a curve connecting $w_i$ with $v_{i+1}$
(highlighted in red).
This curve is at least as long as $\lVert w_iv_{i+1} \rVert$, which is longer than~$e$.
Taking the convex hull can only make the curve even shorter.
Thus, $S$ cannot be optimal.
See \autoref{fig:struct-MST} for an illustration. 


\begin{figure}
    \centering
    \includegraphics[page=2]{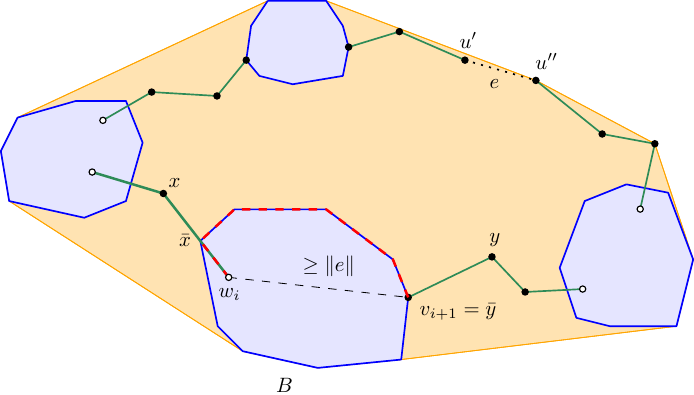}
    \caption{The situation for \autoref{lem:struct-MST}. Adding $e$ and then
    taking the orange region (that is the convex hull of 
    the union of the blue blobs and green edges) 
    as a single blob
    gives a shorter solution, as the red dotted path between $x$ and $y$ is now inside the blob and does not count towards the length of the solution.}
    \label{fig:struct-MST}
\end{figure}


%

So far we showed that in an optimal solution, any MST edge is either a tree-edge or it is contained in a blob. It remains to show that any blob in an optimal solution is a convex hull of a subtree of the MST (as opposed to being a convex hull of a set of points that is disconnected in the MST).

For contradiction, suppose there exist two points $u, v$ that belong to the same blob $B$ and the path from $u$ to $v$ in the MST does not lie inside $B$.
Let $x\bar x$ be the MST edge on this path that leaves $B$ for the first time.
If we remove this edge from $S$, $\bar x$ is still connected to $x$:
First, there is the MST path
from $\bar x$ to $v$, and every edge on this path is either a tree-edge, or its
endpoints are in the same blob. Secondly, $v$ is in the same blob as $x$
Therefore the solution $S$ was not optimal.
%
\end{proof}

\section{Proof of \autoref{lem:enpoints_r_or_l}}
\begin{proof}
    A line segment \(xy\) could intersect the boundary separating the
    left and right area twice. However, as the angle between 
    \(ab\) and the vertical downward ray from \(a\) is obtuse, the lens defined by the disks with radius \(\Vert xy\Vert\) centered at \(x\) and \(y\) always contains \(a\), implying that \(xy\) is not an edge of the MST and thus not contained in \(X_{ab}\). See \autoref{fig:left-right}b.
\end{proof}

\section{Proof of \autoref{lem:valid}} 
\begin{figure}
\centering
    \includegraphics[page=3]{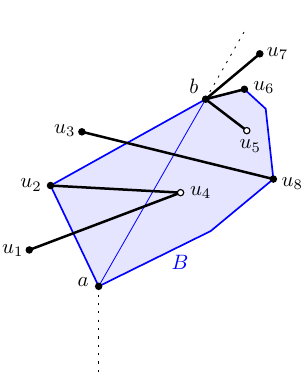}%
    \caption{A few endpoints of edges $X_{ab}$ for a chord $ab$ of a 
    blob~$B$.
    }
    \label{fig:left-right2}
\end{figure}
 \begin{proof}
Suppose $ab$ is a chord of some blob $B$ in an optimal solution.
First observe that, by \autoref{coro:connected-inside-blob} (c) for any edge $uu'\in X_{ab}$, at least one of $u$ and $u'$ must lie in $B$. 

1)
If an endpoint $u$ of an edge $uu'\in X_{ab}^+$ lies outside $B$ (such as
$u_1$, $u_3$, or $u_7$
in Figure~\ref{fig:left-right2}),
it cannot be connected to any other endpoint $v$ of an edge $v\bar v\in X_{ab}^+$, no matter how the two points are 
labeled. 
The reason is that the partner $u'$ of $u$ is in $B$, and either $v$ or its partner $\bar v$ is in $B$. So we would have an MST path that starts in $B$ (at~$u'$), goes out of $B$ (at $u$) and comes back into $B$ (ending in $v$ or $\bar v$), contradicting \autoref{coro:connected-inside-blob} (a).

2)
If a left endpoint $u$ and a right endpoint $v$ of two edges in $X^+_{ab}$ lie both
in $B$, (such as $u_2$ and $u_5$
in Figure~\ref{fig:left-right2}),
they cannot be connected in $T\backslash X^+_{ab}$.
The reason is that the MST must connect $u$ and $v$ within $B$,
by \autoref{coro:connected-inside-blob} (b),
and this is impossible without cutting the chord $ab$ or going through $a$ or $b$.
 \end{proof}
 
\section{Classification of triangles and digons}
\label{sec:classification}
Let \(abc\) be a counterclockwise triangle with lowest point \(a\) and valid chords \(ab\), \(ac\).
Note that \(a\) lies on the left 
side of \(\overline{bc}\).
Let \(O_{abc}\) be the region defined by the intersection of the right 
part of \(\overrightarrow{bc}\) with the wedge defined by the rays \(ab\) and \(ac\), see the orange area in \autoref{fig:classify_sides}.
The \defn{entry edges} for \(abc\) are the entry edges for $bc$ plus the MST edges
$uu'$ with $u'\in\{b,c\}$ and \(u\in O_{abc}\).
The \defn{exit edges} for \(abc\) are the exit edges for $bc$ plus the MST edges
$uu'$ with $u\in\{b,c\}$ and \(u'\in O_{abc}\).

We call \(abc\) an \defn{LR-triangle}
(for ``left-to-right'' triangle), if \(ab\) and \(ac\) are right-facing and there are no exit edges for \(abc\). 
Analogously, \(abc\) is an \defn{RL-triangle} if \(ab\) and \(ac\) are left-facing chords
and there are no exits edges for \(abc\).
The triangle \(abc\) is an \defn{exit-triangle}, if \(ab\) is left facing, \(ac\) is right-facing, and in addition, there is exactly one exit edge for \(abc\).
LR-triangles, RL-triangles, and exit-triangles are \defn{valid triangles.} 
For a valid triangle, let \(X^+_{abc}=X^+_{ab} \cup X^+_{ac} \cup X^+_{bc}\).




\begin{figure}
    \centering
  \includegraphics[page=5,scale=1]{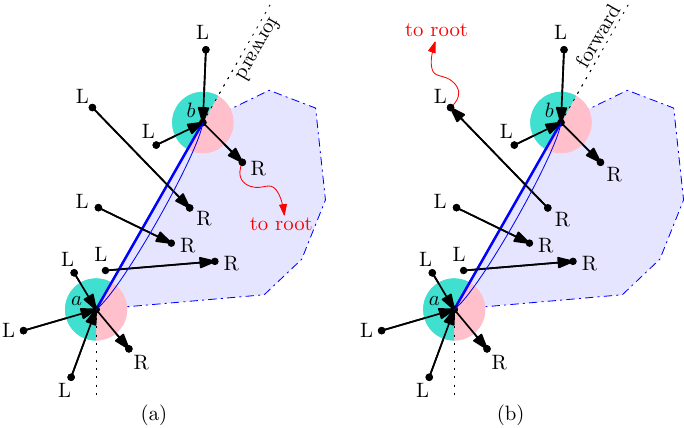}%

    \caption{A left entry digon (a) and a left exit digon (b), together with
    the outline of a 
    blob to which these digons might belong.
    (In (b), the edge from $a$ to the point marked R shows
    that this digon can never be part of a valid blob.
    However,
    this is not tested when considering the segment $ab$:
    This edge will preclude the existence
    of a matching right entry digon.)
    }
    \label{fig:classify_chords}
\end{figure}

Digons are more tricky to define, because they must simultaneously
fulfill the conditions of valid chords and entry or exit walls.
A valid chord
$ab$ with $a\ne r$ is a \defn{valid entry digon} if the
directed edges $uu' \in X^+_{ab}$ satisfy
the following condition:
For every edge 
\(uu'\) in \(X^+_{ab}\), $u=a$ or $u$ is on the backward side,
and $u'=b$ or  \(u'\) is on the forward side.
It is a \defn{left entry digon} if \(ab\) is right-facing
(see \autoref{fig:classify_chords})
and a \defn{right entry digon} if it is left-facing.

If there is exactly one edge
\(uu'\in X^+_{ab}\)
for which the above condition is fulfilled, then $ab$
is a \defn{valid exit digon}.
It is a \defn{right exit digon} if it is right-facing
and a \defn{left exit digon} if it is left-facing
(see \autoref{fig:classify_chords}).
(If \(|X^+_{ab}|=1\) then it is simultaneously an entry digon and an
exit digon, but of opposite sides.)


The case $a=r$ is special:
By definition, all chords \(rb\) are right-facing.
Any 
chord \(rb\) for which all edges in $X_{rb}$
cross
from the right to the left side is a \emph{valid root digon}.
In order to handle root digons consistent with other digons, we assume that a root digon is a right exit digon with a virtual exit edge \(rr'\) with \(\Vert rr' \Vert =0\).
\todo[inline]{K: 
I added the part about the root digon being considered a right exit digon. All other proofs should now work}

\section{Proof of \autoref{lem:exit}}

\begin{proof}
\begin{enumerate}[(a)]
\item 
We look at the exit edges from $B$, i.e., the MST edges $uu'$
with $u'$ not in~$B$
that cross a wall of $B$ or that emanate from 
a vertex $u$ of $B$ into the direction outside~$B$.
We claim that $B$ has at most one exit edge, and it has one exit edge if and only if $B$ does not contain the root.

Suppose $B$ has two exit edges $uu'$ and $vv'$.
We know from \autoref{coro:connected-inside-blob} (b) that $u$ and $v$ are connected by a tree path within $B$, i.e., without going through $u'$ and $v'$. But this is inconsistent with the orientation of $uu'$ and $vv'$.
The same contradiction is obtained 
if $B$ contains the root and has an exit edge. 
Of course, if $B$ does not contain the root, there must be at least one exit edge.
Thus, the claim has been established.

If $B$ contains the root, it is the lowest point $a$ of the blob, by our assumption the right digon is a right exit digon and and none of
the other digon and triangles are exit digons or triangles. 

Otherwise, 
there is a unique
exit edge.
If this edge cuts a wall, this defines a unique exit triangle or digon.
If the exit edge emanates from a vertex of~$B$, the partition of the outgoing
directions into disjoint sectors (\autoref{fig:parition}) ensures that
our convention of assigning the exit tree edge to a triangle (\autoref{fig:classify_sides}) or digon
identifies exactly one triangle or digon as the exit triangle or digon.

\item We have already shown that all chords are valid (\autoref{lem:valid}).
We still need to show that all chords to the right of the exit triangle or digon are left-facing, and all chords to the left
of the exit triangle or digon are right-facing.
If the root is in $B$, then, by convention, all chords are right-facing, and this is the desired result.

Consider now the case that the root is not in $B$, and let $uu'$ be the exit edge. The point $u$ lies inside $B$.
If $u=a$, one can check directly that all chords are oriented correctly, using the edge $uu'\in X_{ab}$.

Consider the case $u\ne a$.
For each chord $ab$, it is sufficient for find \emph{one} edge of $X_{ab}$ whose endpoint is labeled correctly as the forward side.
All chords $ab$ that are crossed by $uu'$ (or touched by $uu'$, in case $u=b$) are oriented correctly, towards $u'$, because the path from $u'$ to the root does not
enter $B$ and therefore does not use edges of $X_{ab}$.

All remaining chords $ab$ are directed towards $u$,
which is seen as follows. Take the tree path $\pi$ from $b$ to the root.
By \autoref{coro:connected-inside-blob} (b), the points in $B$ are connected by tree edges, and hence the path must leave $B$ via the edge $uu'$. Thus, the last vertex of $X_{ab}$ that is visited by $\pi$ must be on the side facing~$u$.
For the chords $ab$ that are not intersected or touched by $uu'$, the direction towards $u$ agrees with the the direction
towards the exit triangle, and hence the proof is completed.
\end{enumerate}
\end{proof}

\section{Interaction between the chords of a valid triangle}
The following lemma gives additional structural insight on how the sides of valid triangles interact.

\begin{lemma} \label{lem:sides_valid_triangles}
    Let \(abc\) be an LR- or RL-triangle and let \(uu'\) be an entry edge of \(abc\). Then the following holds:
    \begin{enumerate}[(a)]
        \item  Every point \(p\) in a left component of \(ac\) lies in a left component of \(ab\). 
        Every point \(p\) in a right component of \(ab\) lies in a right component of \(ac\).
        \item   No point in the left component of \(ac\) or in the right component of \(ab\) lies in the same connected component of  \(T \setminus X^+_{abc}\) as \(u\).
    \end{enumerate}
\end{lemma}

\begin{proof}
\begin{enumerate}[(a)]
    \item 
    \begin{figure}
\centering
\includegraphics[page=2]{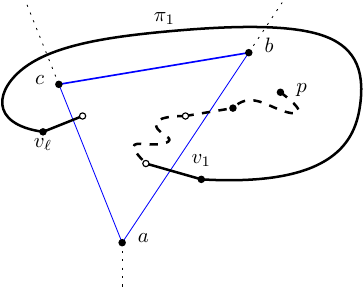}
\caption{ The path \(\pi_1\) connectes a right endpoint in \(X^+_{ab}\) with a left endpoint in \(X^+_{ac}\).}
\end{figure}
    We show the statement for the case of LR-triangles and left components.
    The remaining three cases are symmetric.
    Let \(abc\) be a LR-triangle. Assume that there is a point \(p\) that is in a left component with regard to \(ac\) but in a right component with regard to \(ab\).
    Let \(\pi\) be the path from \(p\) to the root.
    
    Consider the subpaths \(\pi\setminus (X^+_{ab}\cup X^+_{ac})\) of \(\pi\) and a subset \(\Pi=\pi_1,\dots \pi_k\) of these paths.
    \(\Pi\) contains the first path of \(\pi\setminus (X^+_{ab}\cup X^+_{ac})\) if it ends in an endpoint of \(X^+_{ac}\) and all paths that connect an end point in \(X^+_{ab}\) to an end point in \(X^+_{ac}\). 
    Note that \(\Pi \neq \emptyset \) as otherwise \(\pi\) would not use an edge in \(X^+_{ac}\).
    As \(ac\) is a right-facing chord, this would imply that \(p\) is in a right component of \(ac\).

    First note that if \(xx'\in X_{ab} \cap X_{ac}\), then \(x\) is a left end point of an edge in \(X_{ab}\) if and only if it is a left end point of an edge in \(X_{ac}\).
    Additionally, let \(e\) be an edge adjacent to \(a\) and \(\overline{a}\) the other endpoint of this edge.
    Then, if \(\overline{a}\) is a right endpoint with regard to \(ab\) it is also a right endpoint with regard to \(ac\).
    Symmetrically, if \(\overline{a}\) is a left endpoint with regard to \(ac\) it is also a left endpoint with regard to \(ab\).
    
    We claim, that all paths in \(\Pi\) contain vertices that are in a left component of \(ac\) and in a right component of \(ab\). Furthermore, each path in \(\Pi\) that does not start with \(p\) starts and ends with a vertex in \(X^+_{ab}\oplus X^+_{ac}\).
    We show this in an inductive fashion, starting with \(\pi_1\).
    If \(\pi_1\) starts with \(p\) the claim follows by the definition of \(p\).
    In the other case let \(\pi_1 = (v_1,\dots, v_\ell)\). 
    Then \(v_1\) is an end point of an edge in \(X^+_{ab}\) and \(v_\ell\) is an endpoint of an edge in \(X^+_{ac}\).
    As the prefix of \(\pi\) that ends with \(v_\ell\) did not use any edges in \(X^+_{ac}\), \(v_\ell\) cannot be a right end point of an edge in \(X^+_{ac}\) since $p$ is in the left component of $ac$.    
    As there are no exit edges in \(abc\), the point \(v_1\) cannot be a left end point of \(X^+_{ab}\).
    This shows that all points on \(\pi_1\) are in a left component of \(ac\) and in a right component of \(ab\).
    By the observation above, they cannot be in \(X^+_{ab}\cap X^+_{ac}\).s
    This argument can be extended inductively to show the claim for all \(\pi_i\in \Pi\).

    The claim implies that  \(\pi_k\) ends with a left endpoint of an edge in \(X^+_{ac}\setminus X^+_{ab}\) or a right end point of \(X^+_{ab} \setminus X^+_{ac}\). 
    If \(\pi_k\) ends with a right end point \(x\) of an edge in \(X^+_{ab}\), then all points on \(\pi\) that are successors of \(x\), in particular \(r\) are in a left component of \(ac\). 
    This is a contradiction to the assumption that \(ac\) is a right facing chord.
    In the other case, \(\pi_k\) ends with a left end point of \(X^+_{ac}\) and the path does not use any edges of \(X^+_{ab}\) anymore. 
    Consider the suffix of \(\pi\) starting at the last point of \(\pi_k\).
    All points on this suffix are in a right component of \(ab\).
    Let \(yy'\) be the last edge of this suffix that uses an edge in \(X^+_{abc}\).
    This edge has to be a left-to-right edge of \(X^+_{ac}\), as \(ac\) is  a right-facing chord. 
    Thus, \(y'\) lies on the right side of \(ac\).
    There are three possible regions on the right side of \(ac\) for \(y'\), all leading to a contradiction.
    If \(y'\) lies on the right side of \(ab\) this would make \(y\) a point in a left component of \(ab\), a contradiction.
    If it lies in \(O_{abc}\) this would make \(yy'\) an exit edge. 
    Finally, if it lies in the triangle, the suffix of \(\pi\) connecting \(y'\) to \(r\) lies completely in \(abc\). 
    This is only possible if \(a\) is the root.
    Consider the last edge \(\overline{r}r\) on \(\pi\). This edge is in \(X^+_{ab}\), a contradiction to the assumption, that the suffix does not use any such edges.

    \begin{figure}
\centering
\includegraphics[page=3]{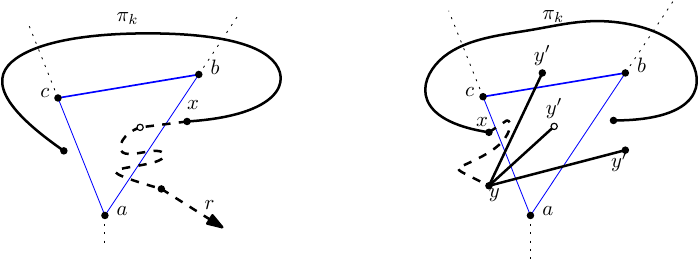}
\caption{(left) if \(\pi_k\) ends with a right endpoint in \(X^+_{ab}\) then \(ac\) is not right-facing. (right) the last edge in \(X^+_{ac}\) ends in one of the three regions, all lead to a contradiction.} 
\end{figure}

    \item W.l.o.g assume that \(p\) is in the right component of \(ab\). We show that $u$ is in the left component of \(ab\).
    
    Let \(\pi\) be the path connecting \(u\) to \(r\) in \(T\). 
    As \(abc\) is not an exit triangle, there existss an edge \(vv'\) such that $v'$ is the first point of $\pi$ after $u$ that is outside of the triangle \(abc\).
    Note that \(vv'=uu'\) is possible. This also means that $vv'$ is either in \(X^+_{ab}\) or in \(X^+_{ac}\).

   We consider two cases based on the position of \(v'\). 
    If \(vv'\in X^+_{ab}\) then \(v\) is in the left component of $ab$ by definition. 

    If \(vv'\in X^+_{ac}\) then \(v'\) is a left endpoint and thus in a left component of \(ac\).
    By part (a) of this lemma, \(v'\) is then also in a left component of \(ab\).
    However the subpath of \(\pi\) between \(u\) and \(v'\) does not go through \(X^+_{ab}\) by assumption.
    Hence, $v$ is also in the left component of $ab$.

    In both cases we get that \(v\) is in the left component of \(ab\). As the path between \(u\) and \(v\) does not use any edges of \(X_{abc}^+\), this implies that \(u\) also in in a left component of \(ab\).
    As \(p\) is in a right component of \(ab\), they cannot be in the same connected component of \(T\setminus X^+_{abc}\).
    The other cases follow similarly.

    \begin{figure}
\centering
\includegraphics[page=5]{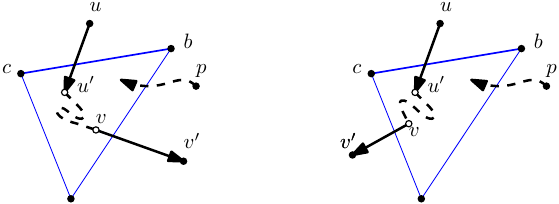}
\caption{Illustration of the proof of Lemma 8 (ii).} 
\end{figure}
    \end{enumerate}
\end{proof}

\section{Relation between subproblems}
\label{sec:subproblem_structure_digon}

Here we describe the relations between the sets of points that define the subproblems. 
This is important to show that the subproblems can be computed in order of increasing sizes of \(V_{ab}\) and \(V_u\). 
\(\Delta_{abc}\) denotes the set of points in the interior of the triangle \(abc\), and
\(W_{abc}\) is the set of points \(w\) such that \(ww'\) is an entry edges for \(abc\). 

\begin{lemma}\label{lem:subproblem_structure_triangle}
    \begin{enumerate}[(a)]
\item
If \(abc\) is an LR-triangle, then 
    \begin{align*}
        V_{ab} =\{c\} \cup V_{ac} \cup \Delta_{abc} \cup \bigcup_{w\in W_{abc}} V_w 
    \end{align*}
    Moreover,  $V_{ab}$ is a strict superset of $V_{ac}$.
\item  If \(abc\) is an RL-triangle, then 
    \begin{align*}
       V_{ac} =\{b\} \cup \vec V_{ab} \cup \Delta_{abc} \cup \bigcup_{w\in W_{abc}} V_w 
    \end{align*}
    Moreover,  $V_{ac}$ is a strict superset of $V_{ab}$.
   \item  If \(abc\) is an exit triangle with exit edge \(u u'\), then
   \begin{align*}
       V_u = \{a,b,c\} \cup V_{ab} \cup V_{ac} \cup \Delta_{abc}  \cup \bigcup_{w\in W_{abc}} V_w
   \end{align*}
    Moreover, $V_u$ is a strict superset of $\overleftarrow V_{ab}$, 
    $\overrightarrow V_{ac}$ and 
    $\bigcup_{w\in W_{abc}} V_w$.
\end{enumerate}
\end{lemma}

\begin{proof}
Note that no path in $T$ joins a point inside $abc$ to a point outside $abc$ without passing through $X^+_{abc}$.


\begin{enumerate}[(a)]
    \item 
    \begin{enumerate}
        \item[\(\subseteq\):] 
        First we show that every point \(p \in V_{ab}\) is at least one of the sets on the right hand side.  
    If \(p=c\) we are done.
    Otherwise, let \(r_p\) be the first vertex of the path \(\pi\) from \(p\) to \(r\) that is an endpoint in \(X^+_{abc}\).
 Note that \(r_p\) exists, as \(p\) is on the backward side of \(ab\) and thus \(\pi\) has to use at least one edge of \(X^+_{ab}\).
 If \(r_p \in \Delta_{abc}\), then \(p\) also lies in \(\Delta_{abc}\) as the path can only leave the triangle by using an edge in \(X^+_{abc}\).
     In the other case, if \(r_p\) is the left endpoint of an edge in \(X^+_{ac}\), then \(p\in V_{ac}\) as it lies in a backward component of this edge.
     Similarly, if \(r_p\in W_{abc}\) the statement follows directly.
     
     \item[\(\supseteq\):]
     Now we show that all sets on the right are subsets of \(V_{ab}\).
     Let \(p\in V_{ac}\).
     Then by definition, \(p\) is in a left component of the chord \(ac\) and in a right component of \(ab\).
     By \autoref{lem:sides_valid_triangles} we get \(p\in V_{ab}\).
     
     Similarly, if \(p\in V_w\) for some \(w\in W_{abc}\) but not in \(V_{ab}\) consider the path \(\pi\) from \(p\) to \(r\).
     Let \(ww'\) be the entry edge defined by \(w\). 
    Let \(C\) be the component of \(T\setminus X^+_{abc}\) that contains \(w\) and \(\pi' = (x',\dots, w)\) the subpath of \(\pi\) in \(C\).
    Then \(x'\) lies outside of the triangle \(abc\) and either \(x'=p\) or \(x'\) is an endpoint for an edge in \(X^+_{ab}\cup X^+_{ac}\).
    By \autoref{lem:sides_valid_triangles}, \(x'\) cannot be in a left component of \(ac\) or a right component of \(ab\). 
    Thus \(x'\) is in a right component of \(ac\) and a left component of \(ab\). If \(x'=p\), then we have \(p\in V_{ab}\).
    If \(x'\neq p\), note that the only endpoints that define such components are in the wedge defined by the rays through \(ab\) and \(ac\).
    As \(x'\) does not lie inside the triangle, \(xx'\) crosses \(bc\) from the inside to the outside, making it an exit edge. 
    This would mean that $abc$ is an exit triangle but $abc$ is an LR-triangle. Hence, we always have $x'=p$ and therefor \(p\in V_{ab}\).

    Finally, let \(p\in \Delta_{abc}\).
    Assume to the contrary that \(p\notin V_{ab}\). Then the first edge  \(yy'\) in \(X^+_{ab}\) on the path \(\pi\) from \(p\) to the root crosses \(ab\) from right to left.
    Thus \(\pi\) has to leave the triangle before \(yy'\). As there are no exit edges, the edge \(xx'\) leaving the triangle is in \(X^+_{ac}\), making \(x'\) a point right of \(ab\) but left of \(ac\), a contradiction to \autoref{lem:sides_valid_triangles}.
    Hence, we get $p \in V_{ab}$.

    To see that \(c\in V_{ab}\), consider the edge \(cc'\) on the path from \(c\) to \(r\).
    If \(cc'\in X_{ab}\), then \(c\in V_{ab}\) by definition. 
    In the other case \(c\) is in the same connected component of \(T\setminus X^+_{ab}\) as \(c'\).
    As \(abc\) is not an exit triangle, \(c'\) either is left of \(ac\), implying \(c'\in V_{abc}\) or in \(\Delta_{abc}\) implying \(c'\in \Delta_{abc}\). In both cases we showed above, that \(c'\in V_{ab}\) and thus also \(c\in V_{ab}\).
    \end{enumerate}

The strict inclusion follows as \(c\notin V_{ac}\) by definition.

\item Analogous.
\item 
\begin{description}
    \item[\(\subseteq\):] 
    First, we show that every point in \(V_u\) is contained in at least one of the sets on the right side of the equation.
Let \(p\in V_u\). If \(p\in\{a,b,c\}\) we are done. In the other case, again let \(r_p\) be the first endpoint of an edge in \(X^+_{abc}\) on the path from \(p\) to the root.
As \(uu'\) is the exit edge of the triangle and \(p\in V_u\), \(r_p\) exists.
If \(r_p\in\Delta_{abc}\), then \(p\in \Delta_{abc}\) as the part of \(\pi\) connecting \(p\) to \(r_p\) does not leave the triangle.
In the other case, \(r_p\) is a right endpoint of an edge in \(X^+_{ab}\), a left endpoint of an edge in \(X^+_{ac}\) or an outward point of an edge in \(W^+_{abc}\).
By the same arguments as in (a) it follows, that \(p\in V_{ab}, p\in V_{ac}\) or \(p \in V_{r_p}\) respectively.
\item[\(\supseteq\):]
Now we show that every point in one of the sets on the right is also contained in \(V_u\).
Let \(p\in V_{ab} \cup V_{ac}\) and let \(xx'\) be the last edge on the path \(\pi\) from \(p\) to \(r\) that is in \(X^+_{ab}\cup X^+_{ac}\).
Then \(x'\) is a right endpoint in \(X^+_{ac}\) or a left endpoint in \(X^+_{ab}\).
Thus \(x \notin X^+_{ab}\cap X^+_{ac}\).
If \(xx'\) is also in \(X^+_{bc}\), then \(xx' = uu'\) as \(uu'\) is the only exit edge of \(abc\).
In the other case \(x'\) is in the interior of \(abc\) and \(\pi\) has to exit the triangle through \(bc\). 
Again, \(uu'\) is the only exit edge and thus \(\pi\) goes through \(u\).

If \(p\in \Delta_{abc}\), the same argument holds, as either the path exits the triangle directly through the exit edge, or it goes through an edge in \(X^+_{ab} \cup X^+_{ac}\) to point \(x\) in \(V_{ab}\cup V_{ac}\). 
Then \(p\) and \(x\) are in the same subtree of \(T\) and thus \(p\in V_{u}\).

Now assume \(p \in \{a,b,c\}\) and let \(pp' \) be the first edge on the path from \(p\) to \(r\).
Then \(pp'\in X^+_{abc}\).
If \(pp' \in X^+_{ab} \cup X^+_{ac}\) or \(p' \in \Delta_{abc}\), then \(p'\in V_{ab}\cup V_{ac}\cup \Delta_{abc}\). 
From the argument above it follows that \(p'\in V_u\) implying that also \(p\in V_u\).
In the other case, \(pp'\in W_{abc}\) and thus \(pp'\) is the exit edge, implying \(p=u\). As \(u\in V_u\) by definition the statement follows.

\end{description}

\(V_{ab}\) and \(V_{ac}\) are strict subsets of \(V_u\) as none of them contain \(a,b\) or \(c\).
Furthermore, we have that \(V_w\subseteq V_u\) for each \(w\in W_{abc}\).
As the  unique outgoing edge of \(w\) is an entry edge and the unique outgoing edge of \(u\) is an exit edge for \(abc\), we have \(u\neq w\) and thus \(V_w \subset V_u\) for all \(w\in W_{abc}\).
\qedhere
\end{enumerate}
\end{proof}

There is an analogous statement for digons: 

\begin{lemma}\label{lem:subproblem_structure_digon}
\begin{enumerate}[(a)]
    \item If \(ab\) is a valid entry digon, then 
    \begin{align*}
        V_{ab} = \bigcup_{\substack{ww' \in X^+_{ab}}} V_w
    \end{align*}
    If \(|X^+_{ab}| = 1\) then \(V_{ab} =V_w\) for the unique backward endpoint \(w\).
    
     \item If \(ab\) is a valid exit digon with exit edge \(u u'\), then
    \begin{align*}
        V_u = \{a,b\} \cup V_{ab} \cup \bigcup_{ww' \in X^+_{ab}\setminus \{uu'\}} V_w
    \end{align*}
    \(V_u\) is a proper superset of each \(V_w\) on the right side.
    \end{enumerate}
\end{lemma}
\begin{proof}
\begin{enumerate}[(a)]
    \item Let \(p\in V_{ab}\) and \(r_p\) be the first endpoint of an edge in \(X^+_{ab}\) on the path from \(p\) to the root.
    As \(ab\) is an entry digon, \(r_p\) lies on the backward side of \(ab\), and thus \(p\in V_{r_p}\).
    If there is more than one edge in \(X^+_{ab}\) then the backwards endpoints do not lie in the same connected components of \(T\setminus X^+_{ab}\) and thus \(V_{ab}\) is a proper super set of all sets \(V_w\).
    \item 

    This proof is similar to \autoref{lem:subproblem_structure_triangle}.
\begin{description}
    \item[\(\subseteq\):] 
    First, we show that every point in \(V_u\) is contained in at least one of the sets on the right side of the equation.
Let \(p\in V_u\). If \(p=a\) or \(p=b\) we are done. In the other case,  let \(x\) be the first endpoint of an edge in \(X^+_{ab}\) on the path $\pi$ from \(p\) to the root.
As \(uu'\) is the exit edge of the digon  and \(p\in V_u\), \(x\) exists. 
Then $x$ is either in the backward or the forward component of $T\backslash X^+_{ab}$.
If it is in the backward component, then $x \in V_{ab}$ and therefore als $p \in V_{ab}$.
If $x$ is in the forward component, let $x'$ be next point on $\pi$. Then $xx'\in X^+_{ab}\backslash\{uu'\}$. Hence, the path from $p$ to $x$ is a subpath of $V_x$. Therefore $p\in V_x$.

\item[\(\supseteq\):]
Now we show that every point in one of the sets on the right is also contained in \(V_u\).

Let \(p\in V_{ab}\) and let \(xx'\) be the last edge on the path \(\pi\) from \(p\) to \(r\) that is in \(X^+_{ab}\).
Then \(x'\) is a point on the forward side in \(X^+_{ab}\) since $p$ is on the backward side. This means that $xx'$ is an exit edge and hence $xx'=uu'$. So $u$ is on $\pi$ and therefore $p \in V_u$.

If $p\in\{a,b\}$, then consider the path $\pi$ from $p$ to $r$. Since $uu'$ is the unique exit edge of the chord $ab$, $\pi$ contains $u$. Hence, $p\in V_u$.

Now let $p \in V_w$ for some $ww' \in X^+_{ab} \backslash \{uu'\}$. If $w'$ is a endpoint on the backward side of $X_{ab}$, then $w'\in V_{ab}\subset V_u$ and $p\in V_u$.
If $w' \in\{a,b\}$, then $w'\in V_u$ by the case above. Hence, $p\in V_u$.
\end{description}

    Note that $u\in V_u$ but $u\notin V_w$ for any $w$ with $ww' \in X^+_{ab}\setminus \{uu'\}$ 
    \qedhere
\end{enumerate}
\end{proof}

\section{Preprocessing}
\label{sec:preprocessing}
\begin{lemma}\label{lem:preprocessing}
    The following preprocessing information can be easily computed in \(O(n^3)\) total time and stored in \(O(n^3)\) space:
    \begin{itemize}
        \item The MST \(T\) rooted at the lowest point~$r$.
        \item For every point $u$, the size of the subtree $V_u$ of the MST rooted at $u$.
\item For every pair of points $a,b$, 
the sets \(X^+_{ab}\) and \(X_{ab}\).
\item For every potential chord $ab$, 
the validity of $ab$ as a chord, the partition of $V\setminus\{a,b\}$ into
backward and forward components, the set $V_{ab}$ of vertices in the backward component, and its size.
\item 
For each ordered pair $(b,c)$ of points, the
classification of $X_{bc}$ 
into entry and exit edges as well as \(W_{bc}\) and \(\bigcup_{w\in W_{ab}} V_w\).
        \(\texttt{wall}[b,c]\)
        \item For every edge \(ww'\) of the MST, the set \(ET_w\) of directed segments \(bc\) such that
        $ww'$ is the unique exiting edge in $X_{bc}$. 
    \end{itemize}
\end{lemma}
\begin{proof}
The MST \(T\) can be computed in \(O(n \log n)\) time~\cite{shamosClosestpointProblems1975}.
A post order-tree traversal directly gives the sets \(V_u\). As the total size of these sets is \(O(n^2)\), reporting all sets takes \(O(n^2)\) time.
To compute \(X^+_{ab}\) and \(X_{ab}\), the MST \(T\) is traversed once for every pair \(a,b\) of points, explicitly testing all edges in \(T\) for containment in the sets. 
For each edge in the set, we store which endpoint is the left and which is the right endpoint.
During the traversal we keep track and store the connected components of \(T\) in \(X^+_{ab}\).
Whenever an edge \(uu'\in X^+_{ab}\) is encountered, the children of \(u\) are in a new connected component.
Each connected component is then traversed. 
Consider the component that contains \(r\). If all leaves in this component are on the same side, then this side will be the forward side. Otherwise, \(ab\) is not a valid chord.
For all other component, check if its root and all leaves are on the same side.
If this is not the case, then \(ab\) is not a valid chord.
In the other case,  assign the component to the side defined by the vertices.
If they are in the backward side, add all vertices of the component to \(V_{ab}\).
This takes \(O(n)\) time for each pair and thus \(O(n^3)\) time overall.

Now consider a directed pair \((b,c)\) of points and iterate over \(X_{bc}\). 
Every edge \(ww'\) with \(w\) left of the directed line through \(bc\) is an entry edge. Add \(w\) to the set \(W_{bc}\).
The remaining edges are exit edges.
If a directed pair \((b,c)\) has only one exit edge, then add \(bc\) to \(ET_w\).
For each entry edge \(ww'\) traverse the subtree of \(T\) rooted at \(w\) to collect \(\bigcup_{w\in W_{ab}} V_w\).
This can be done by one traversal of \(w\), if the vertices \(w\) are stored in  \(W_{bc}\) in a traversal order on the tree.
Thus again, considering one directed pairs \((b,c)\) takes \(O(n)\) time, and the overall time of \(O(n^3)\) follows.
\end{proof}

\section{Solution of chord and edge problems}
\label{sec:solution}
The chord and edge subproblems are solved 
as follows:
\subparagraph{{Chord} problems.}
Let \(ab\) be a valid chord.
Assume
without loss of generality  that it is right-facing. 
For the case that \(ab\) is an interior chord,
we have to consider all
\(c\) such that \(abc\)
is a counterclockwise triangle.
We can check whether $abc$ is a valid LR-triangle
by checking that (a) $ac$ is a valid right-facing chord,
(b) the wall $bc$ has no exit edges, and in addition, (c)~there are no outgoing MST edges from $b$ or $c$ into the
region $O_{abc}$.
Conditions~(a) and (b) have been tested during preprocessing, and
the last condition can be checked in constant time
 because the
degree of a Euclidean MST is bounded by~6.

Then:
The optimum length of the solution that has \(abc\) as a valid LR-triangle in the blob is the sum of the length of all problems $\texttt{wall}[bc]$ plus the length of all problems \(\texttt{edge}[w]\) for \(w\in W_{abc}\setminus W_{bc}\) plus the length of \(\texttt{chord}[a,c]\) plus length of the side \(bc\).

The other possibility is that $ab$ forms a left entering digon:
This possibility is valid if all edges in \(X_{ab}\) cross from the backwards to the forward side. 
Then the solution of the subproblem \(\texttt{chord}[a,b]\) is the sum of \(\texttt{edge}[u]\) for \(uu' \in X_{ab}\) plus the length of the line segment \(ab\).

The value \(\texttt{chord}[a,b]\) is then the minimum value over all \(LR\)-triangles and the digon case. 


\subparagraph{{Edge} problems.}
Let \(u \in P\) and \(uu'\) be the edge in the MST that is defined by $\texttt{edge}[u]$.
Then in any solution, \(u\) is either incident to only tree-edges or it connects a blob. 
In the first case, the length of the solution is simply \(\Vert uu' \Vert\) plus the value of \(\texttt{edge}[v]\) for all children of \(v\) in the MST.
In the second case, there is an exit triangle or exit digon with \(uu'\) being the unique exiting edge. 
Then the length of the solution is the minimum over all valid exit triangles \(abc\) with exiting edge \(uu'\).
Let \(abc\) be a fixed exit triangle such that \(uu'\) is the exit edge.
Then the size of the solution is the \(\Vert bc\Vert + \texttt{chord}[a,b] + \texttt{chord}[a,c] + \sum_{w\in W_{abc}} \texttt{edge}[w]\).
For an exit digon \(ab\) with exit edge \(uu'\) the size of the optimal solution is \(\Vert ab\Vert\) added to the sum of \(\texttt{edge}[w]\) for all \(w \in X_{ab}\) where \(w\) is on the forward side of \(ab\).

The problem for the root can be considered as a degenerate edge problem with a virtual predecessor \(r'\) with \(\Vert rr'\Vert = 0\).
\todo{I am not sure if this is the correct way of phrasing it. Maybe we need a different wording for exit digons?}

\section{Runtime analysis, proof of \autoref{thm:algo}}
\label{sec:runtime}
    For the running time, note that by \autoref{lem:preprocessing} the MST and all relevant information can be preprocessed in \(O(n^3)\) time.
    There is a total of \(O(n^2)\) subproblems of the type \(\texttt{chord}[a,b]\). For each the at most \(O(n)\) elements $c$ in \(V_{ab}\) are traversed in \(O(n)\) total time.
    As the degree of any vertex in the MST is at most \(6\), only \(O(1)\) \texttt{edge} problems are considered in condition to one \texttt{chord} and one \texttt{wall} problem.\todo{This is repeated somewhere else in more detail.}
    Thus the overall time for the \texttt{chord} problems is \(O(n^3)\).
    For the \(O(n^2)\) wall problems, each takes \(O(n)\) \texttt{edge}-subproblems, for an overall of \(O(n^3)\) time.

    On the other hand, there are \(O(n)\) subproblems \(\texttt{edge}[u]\).
    For each of them, the first case of the recurrence uses \(O(n)\) smaller subproblems. 
    Here we take the minimum over all candidate walls where \(uu'\) is an exit edge.
    This set is precomputed by \autoref{lem:preprocessing}.
    Aditionally, we consider the \(O(n)\) candidate walls where one endpoint is adjacent to \(u\).
    As each wall can only have one edge exiting, over all subproblems \(\texttt{edge}[u]\), each wall is considered for exactly one \texttt{edge} problem.
    For each candidate wall, the \(O(n)\) possible vertices \(a\) are considered. 
    The recurrence again only refers to \(O(1)\) subproblems.
    Thus, \(O(n^3)\) total time is needed for all \texttt{edge} problems.

\tableofcontents
\end{document}